\documentclass[12pt]{amsart}

\usepackage{amsmath,amssymb}
\usepackage{graphicx}

\makeatletter
\@addtoreset{equation}{section}
\makeatother

\usepackage{enumerate}

\def\ii{{\sqrt{-1}}}
\def\ee{\mathrm e}

\def\CC{\mathbb C}
\def\ZZ{\mathbb Z}

\def\SL{\mathrm {SL}}
\def\hn{\hat{n}{}}
\def\fA{\mathfrak {A}}
\def\fa{\mathfrak {a}}

\def\CC{{\mathbb C}}
\def\RR{{\mathbb R}}
\def\ZZ{{\mathbb Z}}

\def\cH{{\mathcal H}}
\def\LA{{\langle}}
\def\RA{{\rangle}}

\newtheorem{proposition}{Proposition}[section]

\newtheorem{lemma}{Lemma}[section]

\def\book#1{\rm{#1}, }
\def\paper#1{\textit{#1}, }
\def\jour#1{\rm{#1}, }
\def\yr#1{({\rm{#1}) }}
\def\byr#1{{\rm{#1} }}
\def\vol#1{\textbf{#1}}
\def\pages#1{\rm{#1}}

\def\publaddr#1{\rm{#1}, }
\def\publ#1{\rm{#1}, }
\def\by#1{{\rm{#1}, }}

\begin{document}

\title{On a commutative ring structure in quantum mechanics}

\author{Shigeki Matsutani}
\date{}

\maketitle

\begin{abstract}
In this article, I propose a concept of the $p$-on which
is modelled on the multi-photon absorptions
in quantum optics. It provides a commutative 
ring structure in quantum mechanics.
Using it, I will give an operator representation of
the Riemann $\zeta$ function.
\end{abstract}

\bigskip

{\centerline{\textbf{PACS numbers:
02.10.De, 
02.10.Hh, 
03.65.-w. 
 }}}

{\centerline{\textbf{2000 MSC: 
82B10, 
81R10, 
11M99.   
 }}}

{\centerline{\textbf{Keywords:
multi-photon absorption, $p$-on, 
}}}
{\centerline{\textbf{
Riemann zeta function, Euler
product expression.
 }}}

\section{Introduction}

The integer appearing quantum mechanics 
basically comes from eigenvalue of an operator,
which merely has the additive structure.
On the other hand, number theory is a study of the integer
as a commutative ring rather than an additive group.
There prime numbers play the crucial roles whereas they
basically have no meaning in an additive group.
However number theory and quantum mechanics sometimes are
connected \cite{M, M1, M2, MO, V1, V2, VVZ}.
I attempted to answer a question why quantum mechanics is connected
with integer theory in \cite{M2, MO}.
This article is one of the attempts. I will explore
 a commutative ring structure in the harmonic oscillator.

In \cite{BK, M2, MO}, it was showed that
the Gauss sum which is a number theoretic object plays
the central roles in an interference phenomenon, the 
fractional Talbot phenomenon \cite{WW}.
As I investigated the algebraic structures behind the connection between
wave physics and number theory, there are $\SL(2, \RR)$ and
$\SL(2, \ZZ) \subset \SL(2, \RR)$ \cite{M2}.
Though it is well-known,
the generating relation of the Lie algebra of the abelian extension
of $\SL(2, \RR)$, the Heisenberg group, is
$\displaystyle{[\frac{d}{d x}, x]=1}$ \cite{LV}
whereas the defining relation of $\SL(2, \ZZ)$ is $p a + q b=1$ of
$\displaystyle{\begin{pmatrix}p & q \\ -b & a\end{pmatrix}} \in 
\SL(2, \ZZ)$.
These relations are essential in quantum mechanics  and number theory
respectively \cite{RS, IR} and also of the connection in the classical 
optical phenomenon \cite{M2}.

On the other hand,
the Heisenberg group and 
the interference phenomena are represented by the Fourier series 
\cite{LV, RS}.  The Fourier series is the representation space 
of an additive group or
the translation group.
The translation group plays crucial roles in
the interference phenomena and thus must be essential
in the connection between
 number theory and wave physics.

In the computation of the (discrete) Fourier transformation,
the algorithm of the fast Fourier transformation is well-known,
which is based upon a commutative ring structure of the Fourier series \cite{T}:
For a composite integer $\ell = p q$ and $k \in \RR$, we have
\begin{equation}
   \exp(\ii k p q ) = \left(\ee^{\ii k p }\right)^q 
   = \left(\ee^{\ii k q }\right)^p .
\label{eq:I-1}
\end{equation}
This commutative ring structure 
is the key structure in the fast Fourier transformation.
I consider that it also plays the crucial roles in the 
connection in the interference phenomenon \cite{BK, MO, M2},
though the property 
(\ref{eq:I-1})
comes from the primitive fact that the set of the integers
naturally has a commutative ring structure.
In other words,
the Fourier series as the representation space of the
additive group brings the  commutative ring structure to
the interference phenomenon and
contributes to the connection between number theory and
quantum mechanics.

Eigenvalue of the creation operator in the harmonic oscillator
is given by non-negative integers which is
 merely given by an additive (semi-)group generated by $0$ and $+1$.
However even for the harmonic oscillator, we may have such a commutative
ring structure based upon the primitive fact.
Indeed, in quantum optics, the multi-photon
absorptions are  observed and play the important roles.
Algebraic structure of two-photon absorptions was studied by
Brif \cite{B}.
In this article, we introduce an operator
$p$-on, which is modelled on $p$-photon absorptions, in order to
introduce the commutative ring structure into  quantum mechanics.
Further we also define a quantum $p$-on operator.

Related to the harmonic oscillator
in quantum statistical mechanics and field theory,
the Riemann $\zeta$ function \cite{Pa},
$$
	\zeta(s) = \sum_{n=1} \frac{1}{n^s},
$$
naturally appears \cite{IZ, C} as shown in the Appendix.
In fact, the Planck's black body problem is related to $\zeta(4)$
\cite{Pl} and the Casimir effect is to $\zeta(-3)$ \cite{C}.

The Riemann $\zeta$ function was studied by Bost and Connes \cite{BC, CM}
in the framework of non-commutative algebra, which corresponds to
quantum statistical mechanics physically speaking.
The Riemann $\zeta$ function has the Euler product expression \cite{Pa},
\begin{equation}
	\zeta(s) = \prod_{p : \mathrm{prime\  number}} \frac{1}{1-p^s},
\label{eq:Ep}
\end{equation}
which plays crucial roles in number theory.
The prime number has special meanings in the expression.
In the paper \cite{BC},
there appeared an operator whose eigenvalues are prime numbers.

One of the purposes of this article is to show its quantum version
of the Euler product expression and a quantum mechanical
 meaning of (\ref{eq:Ep}) in the harmonic oscillator.
In other words, in this article, I will show that
even in harmonic oscillator whose eigenvalues are mere integers
as an additive semi-group, there are expressions related to
the Euler product expression (\ref{eq:Ep})
if we handle $p$-on and quantum $p$-on. In Discussion I mention
that a quantum Euler product expression of the Riemann $\zeta$ function
 might be related to the absolute derivation \cite{KOW}.
\bigskip

\section{p-on}

\bigskip

The harmonic oscillator in quantum mechanics
provides the integer as its eigenvalue of the eigenstates \cite{Di}.
The harmonic oscillator is given by the Hamiltonian,
$$
H = \frac{1}{2}( a^\dagger a + a a^\dagger),
$$
using  the creation operator $a^\dagger$ and the annihilation operator 
$a$ which satisfy the canonical communication relations,
\begin{equation}
   [a, a^\dagger] = 1, \quad
   [a^\dagger, a^\dagger] = [a, a] = 0. \quad
\label{eq:a_a}
\end{equation}
Let the vacuum states be denoted by $|0\RA$ and $\LA0|$
, {\it{ i.e.}}, $a |0\RA = 0$ and $\LA0| a^\dagger = 0$.
Let the infinite dimensional $\CC$ vector space 
generated by  $a^\dagger$ $(a)$ be denoted by $\fa^+$ $(\fa^-)$,
{\it{i.e.}}, $\fa^+:= \CC[[a^\dagger]]|0\RA$ 
($\fa^-:= \LA0|\CC[[a]]$), where 
$\CC[[a^\dagger]]$ ($\CC[[a]]$) is a commutative
formal expansion algebra of $a^\dagger$ $(a)$,
{\it{i.e.}}, $f = \sum_{n=0}^{\infty} c_n {a^\dagger}^n$, 
$c_n \in \CC$.
The number state in $\fa^+$ and $\fa^-$ given by 
$$
	\frac{1}{\sqrt{n!}}(a^\dagger)^n |0\RA = |n\RA , \quad
	\LA 0 | \frac{1}{\sqrt{n!}}a^n = \LA n |, 
$$
which satisfies the orthonormal relation 
$$
	\LA m |n\RA =  \delta_{n,m}.
$$
Thus a subspace of $\cH:=(\fa^+, \fa^-, \LA | \RA)$ becomes the Hilbert space.
The number operator $\hn:=a^\dagger a$ picks out
 an integer $n$ as its eigenvalue,
$$
	a^\dagger a  |n\RA =   n |n\RA .
$$
These $a^\dagger$, $a$ and $\hn$ obey the relations
\begin{equation}
[\hn, a^\dagger]=a^\dagger, \quad	
[\hn, a]= -a. \quad	
\label{eq:n_a}
\end{equation}

In number theory, the set of integers is studied 
as a commutative ring rather than an discrete additive group.
The eigenvalue of the harmonic oscillator is a mere additive
semigroup because
$a^\dagger|n\RA= \sqrt{(n+1)} |n+1\RA$ or $a^\dagger$ generates $+1$ action
on the state $|n\RA$.

On the other hand in  quantum optics, 
multi-photon absorption, such as two-photon absorption, is known as an
important phenomenon \cite{B,L}.
I show that this phenomenon brings a commutative ring structure
into the harmonic oscillator.

The two-photon absorption occurs by the composite
operator $a_2^\dagger:= (a^\dagger)^2$ such that
$a_2^\dagger|n\RA=
\sqrt{(n+2)(n+1)} |n+2\RA$. 
Brif investigated the quantum system
governed by these composite operator \cite{B}.
He studied the Lie algebra given by the relations
among $(\hn, a_2^\dagger, a_2:= a^2, a^\dagger, a, 1)$.
Besides  (\ref{eq:a_a}) and (\ref{eq:n_a}), they obey
\begin{equation}
[a_2, a_2^\dagger]=4 \hn + 2, \quad	
[\hn, a_2^\dagger]=2 a_2^\dagger, \quad	
[\hn, a_2]= -2a_2, \quad
[a_2, a_2^\dagger]=4 \hn + 2. \quad	
\label{eq:n_a2}
\end{equation}
Brif investigated its representation space precisely.
Further we note the relations,
$$
a_2 a^\dagger_2 = \hn(\hn - 1) , \quad
a^\dagger_2 a_2 = (\hn + 1)(\hn + 2) . \quad
$$
Similarly, we have
relations among
 $(\hn, a_3^\dagger:={a^\dagger}^3, a_3:= a^3, a^\dagger, a, 1)$.
\begin{equation*}
[a_3, a_3^\dagger]=9 \hn^2 + 9 \hn + 6, \quad	
[\hn, a_3^\dagger]=3 a_3^\dagger, \quad	
[\hn, a_3]= -3a_3, \quad
\end{equation*}
$$
a_3 a^\dagger_3 = \hn(\hn - 1)(\hn - 2) , \quad
a^\dagger_3 a_3 = (\hn + 1)(\hn + 2)(\hn +3) . \quad
$$

Such observations show us that
the composite operator $a_p^\dagger:= (a^\dagger)^p$ 
($a_p:= a^p$) is natural.
We call it $a_p^\dagger$ $p$-on when $p$ is a prime number.
For example, for $n=p \cdot m$, we have a relation,
$|n\RA = (a^\dagger)^n|0\RA/\sqrt{n!} = (a_p^\dagger)^m|0\RA/\sqrt{n!}$.
This means that the individual
monomial $(a^\dagger)^{n + m}$ should be regarded
as a commutative ring, {\it{i.e.}},
$(a^\dagger)^n (a^\dagger)^m =(a^\dagger)^m (a^\dagger)^n$.
For later convenience, we also write 
$a_m^\dagger:= (a^\dagger)^m$, ($a_m:= a^m$),
and sometimes call it $m$-on though it may be, a little bit, overuse.

Then we have the following relations;
\begin{proposition}
$$
a_\ell a_\ell^\dagger = (\hn+1)(\hn+2) \ldots (\hn+\ell), \quad
a_\ell^\dagger a_\ell = \hn(\hn-1) \ldots (\hn-\ell+1), \quad
[\hn, a_\ell^\dagger]=\ell a_\ell^\dagger, \quad	
[\hn, a_\ell]= -\ell a_\ell. \quad
$$
\end{proposition}
\begin{proof}
They are proved by the induction.
We have $\ell = 1, 2, 3$ cases.
Let us show $[a, {a^\dagger}^\ell]=\ell {a^\dagger}^{\ell-1}$
because
$$
[a, a_\ell^\dagger]= [a, {a^\dagger}^\ell]=
a^\dagger [a, {a^\dagger}^{\ell-1}] + 
[a, {a^\dagger}] {a^\dagger}^{\ell-1}.
$$
Thus $[\hn, {a^\dagger}^\ell]$ is computed.
The first formula is obtained by
\begin{equation*}
\begin{split}
a_\ell a_\ell^\dagger &= a^\ell{a^\dagger}^\ell \\
&= a^{\ell-1}(a {a^\dagger}^{\ell-1} )a^\dagger\\
&= a^{\ell-1}({a^\dagger}^{\ell-1} a + (\ell-1){a^\dagger}^{\ell-1} )a^\dagger\\
&= a^{\ell-1}{a^\dagger}^{\ell-1} (a a^\dagger + (\ell-1) )\\
&= a^{\ell-1}{a^\dagger}^{\ell-1} (\hn + \ell).
\end{split}
\end{equation*}
Similarly we have the relations of $a_\ell$.
\end{proof}

\bigskip
Let us consider the relations to the Riemann $\zeta$ function.
As formal expressions, we have
$$
 \ee^{a^\dagger} -1 =
 a^\dagger
+ \frac{1}{2!}(a^\dagger)^2
+ \frac{1}{3!}(a^\dagger)^3
+ \frac{1}{4!}(a^\dagger)^4
\cdots, \quad
 \frac{a}{1-a} =
 a
+ (a)^2
+ (a)^3
+ (a)^4
\cdots.
$$
By letting (see the Appendix),
$$
	\LA n| (a^\dagger a)^{-s}  |m\RA
  := \int_0^{\infty} \frac{d \beta}{\beta} \beta^s
	\LA n| \ee^{-\beta (a^\dagger a)}  |m\RA,
$$
  the following proposition holds:
\begin{proposition}
The Riemann $\zeta$ function is expressed by
$$
	\zeta(s) = \LA0| \frac{a}{1-a} (a^\dagger a)^{-s}
	 (\ee^{a^\dagger}-1) |0\RA.
$$
\end{proposition}

\begin{proof}
Due to the independence of each state, we have
$$
	\LA0|a^{\ell} (a^\dagger a)^{-s}
	 (a^\dagger)^m |0\RA = m! \delta_{n,m} m^{-s},
$$
and then the relation is obtained.
\end{proof}

Let $\wp$ be the set of the prime numbers.
Using $p$-on, 
 the Euler product expression (\ref{eq:Ep}) 
is expressed by the following 
proposition;

\begin{proposition} {\rm{(Euler product expression)}}
\label{prop:EP}
$$
	\zeta(s) = \prod_{p \in \wp} \zeta_p(s), \quad
\zeta_p(s)= \frac{1}{p!}
\LA0| \left(a_p \frac{1}{1- (a^\dagger a)^{-s}} a_p^\dagger\right) |0\RA.
$$
\end{proposition}

\begin{proof}
For a prime number $p$, we have
$
	\LA0|a_p (a^\dagger a)^{-\ell s}
	 a_p^\dagger |0\RA = p! p^{-ms}
$
and
\begin{equation*}
\begin{split}
	\LA p| 
\frac{1}{1-(a^\dagger a)^{-s} } |p\RA
  &= \int_0^{\infty} \frac{d \beta}{\beta} \frac{1}{1-\beta^s}
	\LA p| \ee^{-\beta (a^\dagger a)}  |p\RA\\
  &= \int_0^{\infty} \frac{d \beta}{\beta} 
(1+\beta^s + \beta^{2s} + \beta^{3s} + \cdots)
	\LA p| \ee^{-\beta (a^\dagger a)}  |p\RA.
\end{split}
\end{equation*}
\end{proof}
\bigskip

\section{quantum $p$-on and quantum Euler product expression}

In this section, I will propose the quantum $p$-on and the quantum 
Euler product expression along the line of the concept of
$p$-on in the previous section.
Let us define operators 
$A_m$ and $A_m^\dagger$ which are elements of endmorphisms
$\fa_+$ and $\fa_-$, {\it{i.e.}},
$$
A_m^\dagger: \fa_+ \to \fa_+, \qquad
A_m: \fa_- \to \fa_-,
$$
by
$$
A^\dagger_n \cdot (a^\dagger)^m |0\RA = 
\frac{m!}{(mn)!} (a^\dagger)^{mn} |0\RA, \quad
\LA 0 | a_m \cdot A_n = \LA 0 | a^{mn}. 
$$

Physically speaking, 
$A_m^\dagger$ is the creation operator which creates $m$ $\ell$-ons
when it acts on $a_\ell^\dagger |0\RA$.

From the definition, we have their multiplicity;
\begin{lemma} 
$$
	{A^\dagger_m}^n = A^\dagger_{m^n} , \quad
	A_m^n = A_{m^n} , \quad
	{A^\dagger_m} {A^\dagger_n} ={A^\dagger_n} {A^\dagger_m}
	 =A^\dagger_{nm}, \quad 
	{A_m} {A_n} ={A_n} {A_m} = A_{mn}. \quad 
$$
\end{lemma}
\begin{proof}
$ {A^\dagger_m}^2((a^\dagger)^\ell)|0\RA 
= 
\frac{\ell!}{(m\ell)!}
A^\dagger_m((a^\dagger)^{m\ell})|0\RA 
=
\frac{\ell!}{(m\ell)!}
\frac{(m\ell)!}{(m^2\ell)!} ((a^\dagger)^{m^2\ell})|0\RA 
= {A^\dagger}_{m^2}((a^\dagger)^{\ell})|0\RA.$ 
\end{proof}
Thus we have the proposition:
\begin{proposition} 
$\fA^+:=\CC[[\{A_p^\dagger\}_{p \in \wp}]]$ are
$\fA^-:=\CC[[\{A_p\}_{p \in \wp}]]$ are commutative rings.
\end{proposition} 

Further we have their properties.
\begin{lemma} \label{lm:3-2}
$$
\left( \prod_{p \in \wp}\frac{1}
{1-A^\dagger_p} \right) a^\dagger 
|0\RA 
= \sum_{n=1}^\infty \frac{1}{n!} |n\RA, \quad
\prod_{p \in \wp}\left( {1-A^\dagger_p} \right) 
 \sum_{n=1}^\infty \frac{1}{n!} |n\RA = |1\RA, \quad
$$
$$
\LA 0 |a \left( \prod_{p \in \wp}\frac{1}
{1-A_p} \right)  
= \sum_{n=1}^\infty  \LA n|, \quad
 \sum_{n=1}^\infty \LA n|
\prod_{p \in \wp}\left( {1-A_p} \right) 
 = \LA1|. \quad
$$
\end{lemma}
\begin{proof}
Noting their commutativity, 
$$
\left(\frac{1} {1-A^\dagger_p} \right) 
= 1+{A^\dagger_p} +{A^\dagger_p}^2 +{A^\dagger_p}^3 \cdots
= 1+{A^\dagger_p} +A^\dagger_{p^2} +A^\dagger_{p^3} \cdots.
$$
Since every integer $n$ is uniquely given by
$n = \prod_{i=1}^{\ell_n} p_i^{r_i}$ for
certain prime numbers
$p_i$ and positive numbers ${r_i}$ $(i=1, \cdots, \ell_n)$,
we have
$$
\prod_{p\in\wp} \left(
 1 + {A^\dagger_p} +{A^\dagger_p}^2 +{A^\dagger_p}^3 \cdots
\right) a^\dagger |0\RA =\sum_{n=1} |n\RA.
$$
On the other hand, we have
$$
\left(1-A^\dagger_p \right) 
\left(\frac{1} {1-A^\dagger_p} \right) 
= 1+{A^\dagger_p} +{A^\dagger_p}^2 +{A^\dagger_p}^3 \cdots
 -\left({A^\dagger_p} +{A^\dagger_p}^2 +{A^\dagger_p}^3 \cdots\right)
=1.
$$
\end{proof}

Hence we have an quantum version of the 
Euler product expression (\ref{eq:Ep}): 
\begin{proposition} {\rm{(quantum Euler product expression)}}
\label{prop:QEP}
$$
	\zeta(s) = 
 \LA0| a 
\left( \prod_{q \in \wp}\frac{1}{1-A_q}\right)
(a^\dagger a)^{-s}
\left( \prod_{p \in \wp}\frac{1}
{1-A^\dagger_p} \right) a^\dagger |0\RA .
$$
\end{proposition}

\begin{proof}
From the definition and Lemma \ref{lm:3-2},
 we have the relations
$$
 (\ee^{a^\dagger} -1)|0\RA 
 = \left(\prod_{p \in \wp} (1+A^\dagger_p + A^\dagger_{p^2} + \cdots )
\right) a^\dagger |0\RA,
$$
and 
$$
\LA0|\frac{a}{1-a} =
\LA0| a\left(\prod_{p \in \wp} (1+A_p + A_{p^2} + \cdots )
\right).
$$
Due to the above expression, $\zeta(s)$ is equal to
$$
\LA0|
a\left(\prod_{p \in \wp} (1+A_p + A_{p^2} + \cdots )
\right) (a^\dagger a)^{-s}
\left(\prod_{p \in \wp} (1+A^\dagger_p + A^\dagger_{p^2} + \cdots )
\right) a^\dagger |0\RA.
$$
The independence of each $p$-on gives the relation.
\end{proof}

Noting
$$
(1-A^\dagger_p)
\left( \prod_{q \in \wp}\frac{1}
{1-A^\dagger_q} \right) a^\dagger 
|0\RA 
= 
\left( \prod_{q \in \wp, q\neq p}\frac{1}
{1-A^\dagger_q} \right) a^\dagger 
|0\RA, 
$$
the $\zeta$ function might be decomposed to 
the $\zeta_p$ function.
Further due to interesting relation,
$$
\frac{1}{m^\ell !}\LA 0 |a_{m^\ell} (a a^\dagger)^{-s} 
 a_{m^\ell}^\dagger |0\RA
=
\frac{1}{m !}\LA 0 |a_{m} (a a^\dagger)^{-s \ell}  a_{m}^\dagger |0\RA
$$
we have the Proposition;
\begin{proposition} {\rm{(quantum Euler product expression II)}}
\label{prop:QEPII}
$$
	\zeta_p(s) = 
 \LA0| a 
\left(\frac{1}{1-A_p}\right)
(a^\dagger a)^{-s}
\left( \frac{1}
{1-A^\dagger_p} \right) a^\dagger |0\RA 
$$
\end{proposition}
\begin{proof}
$$
 \LA0| a 
\left(\frac{1}{1-A_p}\right)
(a^\dagger a)^{-s}
\left( \frac{1}
{1-A^\dagger_p} \right) a^\dagger |0\RA 
=
 \LA0| a_p
 \frac{1}{1- (a^\dagger a)^{-s}}
a_p^\dagger |0\RA 
= \zeta_p(s).
$$
\end{proof}
The above relation means 
\begin{equation}
\begin{split}
\zeta_p(s) &= 
 \LA0| a 
\left(1+A_p +A_{p^2} +A_{p^3} + \cdots \right)
(a^\dagger a)^{-s}
\left( 1+A^\dagger_p +A^\dagger_{p^2}
+A^\dagger_{p^3} \cdots \right) a^\dagger |0\RA \\
&= \LA0|\left( a+ a_p +a_{p^2} +a_{p^3} + \cdots \right)
(a^\dagger a)^{-s}
\left( a^\dagger_p +a^\dagger_{p^2}
+a^\dagger_{p^3} \cdots \right)  |0\RA \\
\label{eq:QEPII}
\end{split}
\end{equation}
\bigskip

\section{Discussion}

First we comments on an identification the quartet 
$$
(\CC[[a, a^\dagger]],
\fa^-, \fa^+,
\LA0|\CC[[a, a^\dagger]]|0\RA), 
$$
as
$$
(\CC[[\frac{d}{dz}, z]], 
\CC[[\frac{d}{dz}]], \CC[[z]], 
\frac{1}{2\pi\ii}\oint \frac{dz}{z} \CC[[\frac{d}{dz}, z]] \cdot 1)
$$
for $z \in \CC P^1$. This identification could be regarded as
a transformation between harmonic oscillator and operators
on the Fourier series.
We should note that $[a_2, a_2^\dagger]$ is regarded as
$\displaystyle{\left[\frac{d^2}{d z^2}, z^2\right]
=4 x\frac{d}{dx} + 2}$, which
may be related to the quadratic differentials on Riemann surfaces
\cite[Chapter VII.2]{FK}.
Further instead of $\CC P^1$,
for example in \cite{MP} for a algebraic curve, {\it{e.g.}},
$y^r = x^s + \lambda_{s-1} x^{s-1} + \cdots + \lambda_0$, at its
infinite point, the local parameter $z_\infty$ behaves like
$$
z_\infty^r = 1/x + O(z_\infty^{r+1}), \quad 
z_\infty^s = 1/y + O(z_\infty^{r+1}).
$$
In other words, $1/x$ and $1/y$ behave like 
$a^\dagger_r$ and $a^\dagger_s$ respectively.
When we consider more general algebraic curves, there naturally
appear relations among 
$(a_{\ell_1}, a_{\ell_2}, \cdots, a_{\ell_k})$.
They are related to nonlinear integrable system and several
physical phenomenon \cite{BBEIM}.
The dynamics of $\{z^\ell\}$  in the orthogonal polynomial
is connected with the integrable system and the random matrix problem
\cite{S};
the random matrix is also connected with $\zeta$ function \cite{Me}.
Thus this interpretation is not trivial and is very natural from
the viewpoint.

Further we note that
for the system $(x, d/dx)$, $A^\dagger$ could be regarded as
$$
	\frac{d^\ell}{dx^\ell} A_m = 
	\frac{d^{\ell m}}{dx^{\ell m}} , \quad 
	A^\dagger_m x^\ell = \frac{\ell!}{(m\ell)!}x^{m\ell}.
$$
Thus the $p$-on picture is natural from the viewpoint of the 
the identification. 

\bigskip

Secondly we give some comments on 
Proposition \ref{prop:QEP} and Proposition \ref{prop:QEPII} 
of quantum Euler product expression.
The quantum Euler product expression
in Proposition \ref{prop:QEPII}
is reduced to the ordinary Euler product expression (\ref{eq:Ep}).
In other words, in the harmonic oscillator problem, the
natural commutative ring structure exists and provides the
relations to the Euler product expression 
 (\ref{eq:Ep}) of the Riemann $\zeta$ function. 
I have a quantum mechanical interpretation of the Euler product expression
and its quantum meaning as a relation to $p$-on.

I should emphasize that even behind the Planck black body problem and
the Casimir effect, these expressions exist.
(\ref{eq:QEPII}) shows that there exist excitations of $p$-on,
$p^2$-on, $p^3$-on and so on for each prime number.
The multi-photon absorption in quantum optics
is a sure sign of these excitations.
It implies one of answers why the quantum mechanics provides 
a connection with number theory and $p$-adic structure
\cite{M, M1, M2, MO, V1, V2, VVZ}.

\bigskip

Further from the definition, we may have the relation
$$
	A_m\cdot (a^\dagger)^n |0\RA =
      \left\{
\begin{matrix}
	 \frac{n!}{(n/m)!}(a^\dagger)^{n/m} |0\RA  & \mbox{ if }  m | n,\\
	 0                        & \mbox{ otherwise }.\\
\end{matrix}
\right.
$$
In other words, for a prime number $p$, we have
$$
	A_p\cdot (a^\dagger)^{n} |0\RA =
      \left\{
\begin{matrix}
	 \frac{(p^\ell m)!}{(p^{\ell-1} m)!}(a^\dagger)^{m p^{\ell-1}} 
       |0\RA   & \mbox{ if }  
     n = p^\ell m, \ (\ell \ge 1, p \not| m), \\
	 0                        & \mbox{ otherwise }.\\
\end{matrix} \right.
$$
This reminds me of the absolute derivation \cite{KOW},
$$
   \frac{\partial}{\partial p} : \ZZ \to \ZZ,
$$
$$
   \frac{\partial}{\partial p} n
= \left\{
\begin{matrix}
	 \ell p^{\ell-1} m  & \mbox{ if }  
     n = p^\ell m, \ (\ell \ge 1, p \not| m), \\
	 0                        & \mbox{ otherwise },\\
\end{matrix} \right.
$$
for a prime number $p$,  which is introduced for
the study of the Riemann $\zeta$ function. 

I believe that 
this commutative ring structure
in the harmonic oscillator is quite interesting
and gives an answer of the question why  number theory is
connected with quantum mechanics and now we interpret the
Euler product expression in the harmonic oscillator.
It should be emphasized that even behind the Planck black body problem and
the Casimir effect, the commutative ring structure and $p^\ell$-on
exist.

\bigskip

{\bf{Acknowledgement}}

I am grateful to John McKay for telling me the reference
\cite{CM} and to A. Vourdas and Hideo Mitsuhashi for helpful comments.

\bigskip
\section[*]{Appendix L-function}

In this Appendix, I will show the interpretation $L$-functions
 and the Riemann $\zeta$ function in number theory
from a statistical mechanical viewpoint.
In statistical mechanics, we consider
partition functions which are generators of
expectations in canonical ensembles.
The partition function for a statistical
mechanical system $A$ is defined by
\begin{gather}
	Z[\beta]:= \sum_{\mbox{all of states $s$ in $A$}} \ee^{-\beta E(s)},
\end{gather}
where $1/\beta$ is a temperature of the
system $A$ and $E(s)$ is an energy of a state $s \in A$.
By using spectral decomposition, we can also express it as
\begin{gather}
	Z[\beta]=
	\sum_{E} \sum_{s\in S_E}
	 \ee^{-\beta E},
\end{gather}
where $S_E$ is a subset of $A$ which has energy $E$.
Further we sometimes rewrite this
\begin{gather}
Z[\beta]= \sum_{E} c_E \ee^{-\beta E},
\label{eq:1}
\end{gather}
where $c_E$ is number of $S_E$, {\it i.e.}, $c_E:=\# S_E$.
(\ref{eq:1}) is called the energy-representation.

Here we note that $S_E$ is equivalent with respect to
the energy $E$. The equivalence means that
there might exist a group $G_E$ which simply transitively acts on $S_E$.
Then the $c_E$ must be an invariance of $G_E$ or
 a group ring $R=\ZZ[G_E]$.
For example, $c_E$ and $S_E$ might be
 related to an identity representation,
$$
	\hat c_E :=\sum_{x\in G_E} x.
$$
In fact for a map $\varphi : \ZZ[G_E] \ni \sum a_i x_i \to
\sum a_i \in \ZZ$, $\varphi(\hat c_E ) = c_E$.

In statistics, we sometimes deal with sequence of
the $n$-th order expectation value
($n$-th moment) instead of the generator itself.
Thus it is natural to introduce an $n$-th moment,
\begin{gather}
\split
	K[s] &:= \int_0^\infty \beta^s Z[\beta] \frac{d \beta}{\beta}\\
             &=\Gamma(s) \sum_E \frac{c_E}{E^{-s}},
\endsplit
\label{eq:4}
\end{gather}
where $s$ is a natural number.
(Here we note 
that existence of the integral (\ref{eq:4}) is asserted by,
for example, (3.9) Theorem in \cite{Du}.)
It is remarked that $s$ is sometimes extended to a complex number
by analytical continuity.
This  $K[s]$ is called the generalized $\zeta$-function which appears
in \cite{EORBZ,M0}.

Let $E$ be parameterized by an integer or $E=n \in \ZZ$.
When $c_n=1$,  we have the Riemann $\zeta$ function \cite{Pa},
$$
	K[s] = \sum_{n=1} \frac{1}{n^s} = \zeta(s),
$$
which is the main theme of this article.

Let $c_n$ satisfy
$$
	c_n = \left\{ \begin{matrix} 
2 & \mbox{for } n \equiv 1,7 \mbox{ mod } 8,\\
0 & \mbox{for } n \equiv 3,5 \mbox{ mod } 8,\\
1 & \mbox{for } \mbox{ otherwise }.\\
\end{matrix}
\right.
$$ 
$$
	K[s] = L(s,\chi) + \zeta(s),
$$
where $L(s,\chi)$ is the Dirichlet characteristics which is given by
\cite{IR},
$$
	L(s,\chi) = \sum_{n=1} \frac{\chi(n)}{n^s},
$$
and
$$
	\chi(n) := \left\{ \begin{matrix} 
1 & \mbox{for } n \equiv 1,7 \mbox{ mod } 8,\\
-1 & \mbox{for } n \equiv 3,5 \mbox{ mod } 8,\\
0 & \mbox{for } \mbox{ otherwise }.\\
\end{matrix}
\right.
$$
Further 
(\ref{eq:1}) is also related to the Gauss sum \cite{IR, BK, MO}
if 
$\displaystyle{\begin{pmatrix} n \\ p \end{pmatrix}}$ and 
$\displaystyle{\beta= \frac{\sqrt{-1}}{p}}$.

\bigskip

\bigskip

{Shigeki Matsutani}

{e-mail:RXB01142\@nifty.com}

{8-21-1 Higashi-Linkan}

{Sagamihara 228-0811 Japan}

\end{document}